\documentclass[10pt,a4paper,onecolumn]{IEEEtran}
\usepackage[lmargin=2cm, rmargin=2cm, tmargin=2cm, bmargin=2cm]{geometry}
\usepackage[utf8]{inputenc}
\usepackage[T1]{fontenc}
\usepackage[brazil]{babel}
\usepackage{fancyhdr}
\usepackage[round,authoryear]{natbib}
\usepackage{graphicx}
\usepackage{colortbl}
\usepackage{indentfirst}
\usepackage{subfig}
\usepackage{amsfonts}
\usepackage{amsmath}
\usepackage{amsthm}
\usepackage{mwe}

\newtheorem{Theorem}{Theorem}

\newtheorem{Assumption}[Theorem]{Assumption}
\newtheorem{Corollary}[Theorem]{Corollary}

\title{A fuzzy feedback linearization scheme applied to vibration control of a smart structure}
\author{Roberta Varela de Albuquerque Herôncio, João Deodato Batista dos Santos,\\
Wallace Moreira Bessa, Aline Souza de Paula, Marcelo Amorim Savi}
\date{}
    
\begin{document}
    
\maketitle

\abstract{

Smart structures are usually designed with a stimulus-response mechanism to mimic 
the autoregulatory process of living systems. In this work, in order to simulate this natural and self-adjustable 
behavior, a fuzzy feedback linearization scheme is applied to a shape memory two-bar truss. This structural 
system exhibits both constitutive and geometrical nonlinearities presenting the snap-through behavior and chaotic 
dynamics. On this basis, a nonlinear controller is employed for vibration suppression in the chaotic smart truss. 
The control scheme is primarily based on feedback linearization and enhanced by a fuzzy inference system to cope 
with modeling inaccuracies and external disturbances. The overall control system performance is evaluated by means 
of numerical simulations, promoting vibration reduction and avoiding snap-through behavior.

}

\section{INTRODUCTION}

The term smart structures and systems has been used to identify mechanical systems that are capable of 
changing their geometry or physical properties with the purpose of performing a specific task. They must be 
equipped with sensors and actuators that induce such controlled alterations. Several applications in different 
fields of sciences and engineering have been developed with this innovative idea, employing some of the so-called 
smart materials. Shape memory alloys (SMAs), piezoelectric materials and magneto-rheological fluids are some of 
the smart materials largely employed in structural systems. 

Specifically, shape memory alloys are being used in situations where high force, large strain, and low frequency 
structural control are needed. SMA actuators are easy to manufacture, relatively lightweight, and able of producing 
high forces or displacements. Self-actuating fasteners, thermally actuator switches and several bioengineering 
devices are some examples of these SMA applications. Aerospace technology are also using SMAs for distinct purposes 
as space savings achieved by self-erectable structures, stabilizing mechanisms, non-explosive release devices, among 
others. Micromanipulators and robotics actuators have been built employing SMAs properties to mimic the smooth motions 
of human muscles. Moreover, SMAs are being used as actuators for vibration and buckling control of flexible structures. 

SMA thermomechanical behavior is related to thermoelastic martensitic transformations. The shape 
memory effect is a phenomenon where apparent plastically deformed objects may recover their original 
form after going through a proper heat treatment. The pseudoelastic behavior is characterized by complete 
strain recovery accompanied by large hysteresis in a loading-unloading cycle \citep{otsuka1}. Fibers of 
shape-memory alloys can be used to fabricate hybrid composites exhibiting these two different but related 
material behaviors. Detailed description of the shape memory effect and other phenomena associated with 
martensitic phase transformations, as well as examples of applications in the context of smart structures, 
may be found in references \citep{lagoudas1,paiva1,machado1,rogers1,shaw1}.

The investigation of SMA structures has different approaches. The finite element method is an important tool to 
this aim. \citet{auricchio1} proposed a three-dimensional finite element model. \citet{lagoudas2} considered the 
thermomechanical response of a laminate with SMA strips. \citet{cava1} considered SMA bars and \citet{bandeira1} 
treated truss structures. The response of SMA beams was treated by \citet{collet1}, which analyzes the dynamical 
response, as well as \citet{auricchio2}. \citet{auricchio4,auricchio3} presented a solid finite element to describe 
the thermo-electro-mechanical problem that is used to simulate different SMA composite applications. Dual kriging 
interpolation has been employed with finite element method in order to describe the shape memory behavior in different 
reports \citep{trochu1}. \citet{masud1}, \citet{bhattacharyya1}, \citet{liew1} are other contributions in this field.

The two-bar truss, also known as the von Mises truss, is an important archetypal model, largely employed to 
evaluate stability characteristics of framed structures as well as of flat arches, and of many other physical 
phenomena associated with bifurcation buckling \citep{bazant1}. The nonlinear dynamics of this system may exhibit 
a number of interesting, complex behaviors. The snap-through behavior, represented by a displacement jump, is a 
classical example of the complexity behind this simple structure. 

The dynamic behavior of the two-bar truss is even richer when material nonlinearities are considered. In particular, 
the present contribution deals with two-bar trusses made from shape memory materials. \citet{savi6} and \citet{savi5} 
presented numerical investigations of this kind of structure by considering different constitutive models to describe 
the thermomechanical behavior of the SMAs. 

Due to its simplicity, feedback linearization scheme is commonly applied in industrial control systems, specially
in the field of industrial robotics. The main idea behind this control method is the development of a control law 
that allows the transformation of the original dynamical system into an equivalent but simpler one \citep{slotine1}. 
Although feedback linearization represents a very simple approach, an important handicap is the requirement of a 
perfectly known dynamical system, in order to ensure the exponential convergence of the tracking error.

Intelligent control, on the other hand, has proven to be a very attractive approach to cope with uncertain nonlinear systems 
\citep{tese,rsba2010,cobem2005,nd2012,Bessa2014,Bessa2017,Bessa2018,Bessa2019,Deodato2019,Lima2018,Lima2020,Lima2021}. 
By combining nonlinear control techniques, such as feedback linearization or sliding modes, with adaptive intelligent algorithms, 
for example fuzzy logic or artificial neural networks, the resulting intelligent control strategies can deal with the nonlinear 
characteristics as well as with modeling imprecisions and external disturbances that can arise.

On this basis, much effort has been made to combine feedback linearization with intelligent algorithms in order to 
improve the trajectory tracking of uncertain nonlinear systems. On this basis, in order to enhance the tracking 
performance, \citet{ijccc2013} used a fuzzy inference system with the state variables in the premise of the rules
to approximate the unknown system dynamics. However, the adoption of all state variables in the premise of the rules 
is a drawback of this approach. As for instance, for higher-order systems the number of fuzzy sets and fuzzy rules 
becomes incredibly large, which compromises the applicability of this technique. 

In this work, to reduce the number of fuzzy sets and rules and consequently simplify the design process, only on variable 
(an error measure) instead of the state variables, is adopted in the premise of the fuzzy rules. A polynomial constitutive 
model is assumed to describe the behavior of the shape memory bars. Although this model is simple and does not present a 
proper description of the hysteretic behavior, it can qualitatively represent the general SMA behavior. This system has a 
rich dynamic response and can easily reach a chaotic behavior even at moderate loads and frequencies \citep{savi5}. 
Numerical simulations are carried out in order to demonstrate the control system performance.

The main goal is the vibration reduction, avoiding some critical responses as snap-through behavior. A linear actuator is 
employed to help this control procedure and therefore, the SMA actuation is not employed for the control purposes. In this 
regard, we are investigating an SMA structure that needs an appropriate control using external actuators. It is important 
to highlight that SMA properties are being used to achieve other goals than control. This situation is common in distinct 
applications that include aerospace systems as self-erectable structures.

\section{DYNAMIC MODEL}

The two-bar truss is depicted in Figure~\ref{Fig1}. This plane, framed structure, is formed by two identical bars, free 
to rotate around their supports and at the joint.

\begin{figure}[htb]
\centering
\includegraphics[width=0.4\textwidth]{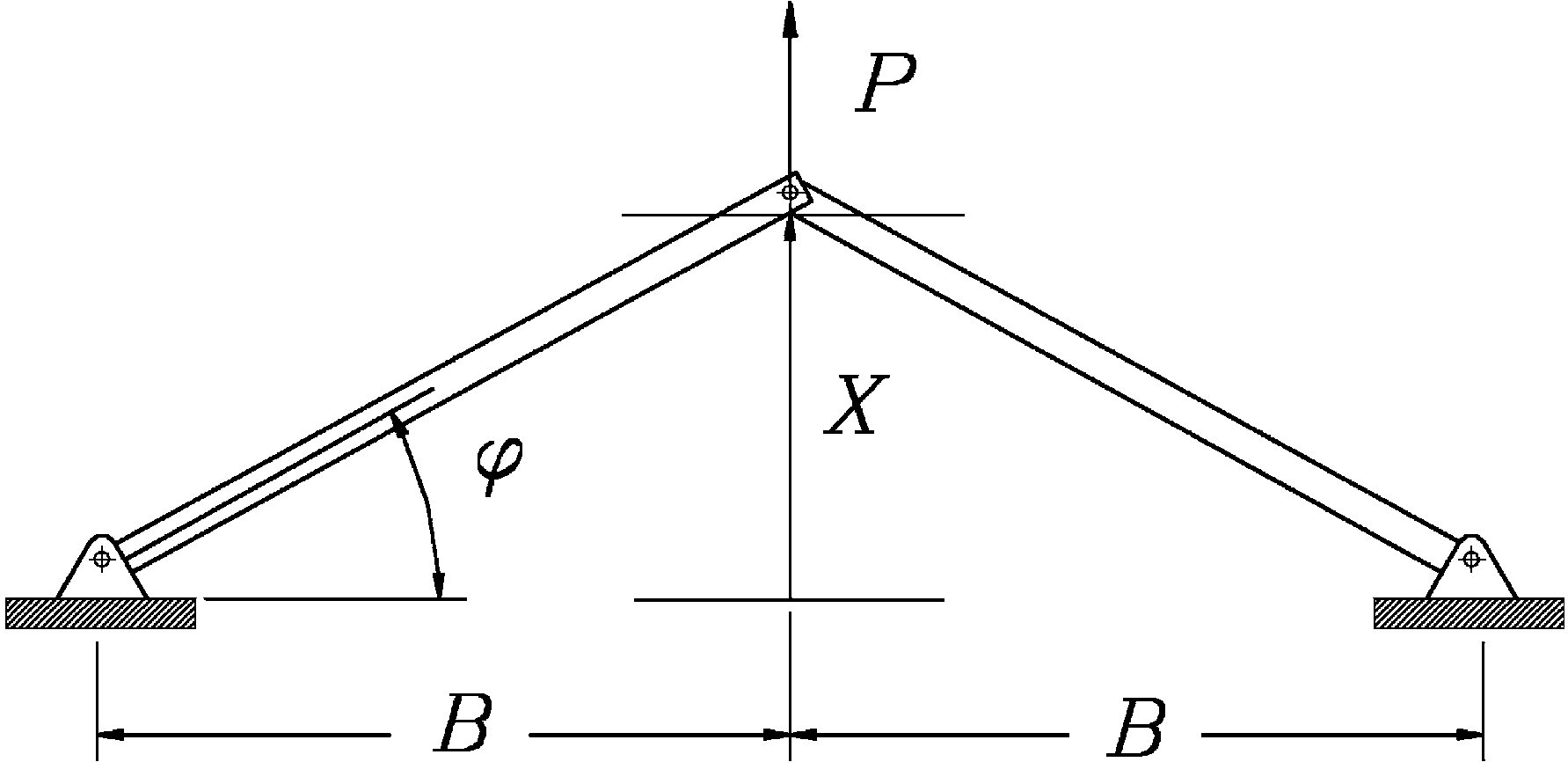}
\caption{Two-bar truss (von Mises truss).}
\label{Fig1}
\end{figure}

In the present investigation, we consider a shape memory two-bar truss where each bar presents the shape memory and 
pseudoelastic effects. The two identical bars have length $L$ and cross-sectional area $A$. They form an angle $\varphi$ 
with a horizontal line and are free to rotate around their supports and at the joint, but only on the plane formed by 
the two bars (Figure \ref{Fig1}). The critical Euler load of both bars is assumed to be sufficiently large so that 
buckling will not occur in the simulations reported here.

We further assume that the structure's mass is entirely concentrated at the junction between the two bars. Hence, the 
structure is divided into segments without mass, connected by nodes with lumped mass that is determined by static 
considerations. We consider only symmetric motions of the system, which implies that the concentrated mass, $m$, can 
only move vertically. The symmetric, vertical displacement is denoted by $X$. Under these assumptions, the dynamic 
behavior is expressed through the following equation of motion

\begin{equation}
-2F\sin\varphi -c\dot{X}+ P = m\ddot{X}
\label{eq:motion1}
\end{equation}

\noindent
where $F$ is the force on each bar, $P$ is an external force and $c\dot{X}$ is a linear viscous damping term used to 
represent all dissipation mechanisms. 

There are several works dedicated to the constitutive description of the thermomechanical behavior of shape memory alloys 
\citep{lagoudas1,paiva1}. In this article, we employ polynomial constitutive model to describe the thermomechanical behavior 
of the SMA bars \citep{falk1,muller1}. Despite the simplicity of this model, it allows an appropriate qualitative description 
of the dynamical response of the system. Its major drawback is the hysteresis description. In this regard, for control purposes,
it should be appropriate with robust controllers that could deal with unmodelled dynamics. Here, dissipation process is 
represented by an equivalent viscous damping term.

Polynomial model is concerned with one-dimensional media employing a sixth degree polynomial free energy function in terms 
of the uniaxial strain, $\varepsilon$. The form of the free energy is chosen in such a way that its minima and maxima 
are respectively associated with the stability and instability of each phase of the SMA. As it is usual in one-dimensional 
models proposed for SMAs \citep{savi4}, three phases are considered: austenite (A) and two variants of martensite (M+, M-). 
Hence, the free energy is chosen such that for high temperatures it has only one minimum at vanishing strain, representing 
the equilibrium of the austenitic phase. At low temperatures, martensite is stable, and the free energy must  have two minima 
at non-vanishing strains. At intermediate temperatures, the free energy must have equilibrium  points corresponding to both 
phases. Under these restrictions, the uniaxial stress, $\sigma$, is a fifth-degree polynomial of the strain \citep{savi4}, i.e.

\begin{equation}
\sigma = a_1(T-T_M)\varepsilon - a_2\varepsilon^3 + a_3\varepsilon^5
\label{eq:stress}
\end{equation}

\noindent
where $a_1$, $a_2$ and $a_3$ are material constants, and $T$ the temperature, while $T_M$ is the temperature 
below which the martensitic phase is stable. If $T_A$ is defined as the temperature above which austenite is 
stable, and the free energy has only one minimum at zero strain, it is possible to write the following condition,

\begin{equation}
T_A = T_M + \frac{1}{4}\frac{a_2^2}{a_1a_3}
\label{eq:condition}
\end{equation}

Therefore, the constant $a_3$ may be expressed in terms of other constants of the material. Now, the following 
strain definition is considered, 

\begin{equation}
\varepsilon = \frac{L}{L_0} - 1 = \frac{\cos\varphi_0}{\cos\varphi} - 1
\label{eq:strain}
\end{equation}

\noindent
with $L_0$ and $\varphi_0$ representing the nominal values of $L$ and $\varphi$, respectively. 

At this point, we can use the constitutive equation (\ref{eq:stress}) together with kinematic equation (\ref{eq:strain}) 
into the equation of motion (\ref{eq:motion1}), obtaining the governing equation of the SMA two-bar truss:

\begin{equation}
\begin{split}
m\ddot{X}&+c\dot{X}+\frac{2A}{L_0}X\Big\{[a_1(T-T_M)-3a_2+5a_3]+[-a_1(T-T_M)+a_2-a_3]L_0(X^2+B^2)^{-1/2}+\\
&+[3a_2-10a_3]\frac{1}{L_0}(X^2+B^2)^{1/2}+[-a_2+10a_3]\frac{1}{L_0^2}(X^2+B^2)+\\
&-\frac{5a_3}{L_0^3}(X^2+B^2)^{3/2}+\frac{a_3}{L_0^4}(X^2+B^2)^2\Big\}=P(t)
\label{eq:motion2}
\end{split}
\end{equation}

\noindent
where $B$ is the horizontal projection of each truss bar (Figure \ref{Fig1}).

Considering a periodic excitation $P=P_0\sin(\omega t)$, equation (\ref{eq:motion2}) may be written in non-dimensional 
form as

\begin{equation}
\begin{split}
x'= &\;y\\
y'= &\;\gamma\sin(\Omega\tau)-\xi y+x\big\{-[(\theta-1)-3\alpha_2+5\alpha_3]+\\
&+[(\theta-1)-\alpha_2+\alpha_3](x^2+b^2)^{-1/2}
-[3\alpha_2-10\alpha_3](x^2+b^2)^{1/2}+\\
&+[-\alpha_2+10\alpha_3](x^2+b^2)
+5\alpha_3(x^2+b^2)^{3/2}-\alpha_3(x^2+b^2)^2\big\}
\label{eq:motion3}
\end{split}
\end{equation}

\noindent
where $\xi$ is a non-dimensional viscous damping coefficient. The dissipation due to hysteretic effect may 
be considered by assuming an equivalent viscous damping related to this parameter. Moreover, the following 
non-dimensional parameters are considered: 

\begin{displaymath}
\displaystyle
\begin{split}
&x=\frac{X}{L},\quad \gamma=\frac{P_0}{mL_0\omega_0^2},\quad \omega_0^2=\frac{2Aa_1T_M}{mL_0},\quad
\Omega=\frac{\omega}{\omega_0},\quad \tau=\omega_0t,\\
&\theta=\frac{T}{T_M},\quad \alpha_2=\frac{a_2}{a_1T_M},\quad \alpha_3=\frac{a_3}{a_1T_M},\quad 
b=\frac{B}{L_0} \quad\mbox{and}\quad (\cdot)'=\frac{\mathrm d(\cdot)}{\mathrm d\tau}
\end{split}
\end{displaymath}

\section{CONTROLLER DESIGN}

\subsection{Feedback linearization}

Consider a class of $n^\mathrm{th}$-order nonlinear systems:

\begin{equation}
x^{(n)}=f(\mathbf{x},t)+b(\mathbf{x},t)u+d
\label{eq:system}
\end{equation}

\noindent
where $u$ is the control input, the scalar variable $x$ is the output of interest, $x^{(n)}$ is the $n$-th time 
derivative of $x$, $\mathbf{x}=[x,\dot{x},\ldots,x^{(n-1)}]$ is the system state vector, $f,b:\mathbb{R}^n\to
\mathbb{R}$ are both nonlinear functions and $d$ is assumed to represent all uncertainties and unmodeled dynamics 
regarding system dynamics, as well as any external disturbance that can arise. 

In respect of the disturbance-like term $d$, the following assumption will be made:

\begin{Assumption}
The disturbance $d$ is unknown but continuous and bounded, \textit{i.\,e}.\ $|d|\le\delta$.
\label{as:limd}
\end{Assumption}

Let us now define an appropriate control law based on conventional feedback linearization scheme that ensures the 
tracking of a desired trajectory $\mathbf{x}_d=[x_d,\dot{x}_d,\ldots,x^{(n-1)}_d]$, \textit{i.\,e}.\ the controller 
should assure that $\mathbf{\tilde{x}}\rightarrow0$ as $t\rightarrow\infty$, where $\mathbf{\tilde{x}}=\mathbf{x}
-\mathbf{x}_d=[\tilde{x},\dot{\tilde{x}},\ldots,\tilde{x}^{(n-1)}]$ is the related tracking error. 

On this basis, assuming that the state vector $\mathbf{x}$ is available to be measured and system dynamics is 
perfectly known, \textit{i.\,e}.\ there is no modeling imprecision nor external disturbance ($d=0$) and the functions 
$f$ and $b$ are well known, with $|b(\mathbf{x},t)|>0$, the following control law:

\begin{equation}
u=b^{-1}(-f+x^{(n)}_d-k_0\tilde{x}-k_1\dot{\tilde{x}}-\cdots-k_{n-1}\tilde{x}^{(n-1)})
\label{eq:lawfl}
\end{equation}

\noindent 
guarantees that $\mathbf{x}\rightarrow\mathbf{x}_d$ as $t\rightarrow\infty$, if the coefficients $k_i$ $(i=0,2,
\ldots,n-1)$ make the polynomial $p^n+k_{n-1}p^{n-1}+\cdots+k_0$ a Hurwitz polynomial \citep{slotine1}.

The convergence of the closed-loop system could be easily established by substituting the control law (\ref{eq:lawfl}) 
in the nonlinear system (\ref{eq:system}). The resulting dynamical system could be rewritten by means of the tracking 
error:

\begin{equation}
\tilde{x}^{(n)}+k_{n-1}\tilde{x}^{(n-1)}+\ldots+k_1\dot{\tilde{x}}+k_0\tilde{x}=0
\label{eq:cl}
\end{equation}

\noindent 
where the related characteristic polynomial is Hurwitz.

The characteristic polynomial could be assured to be Hurwitz by defining $\mathbf{k^\mathrm{T}\tilde{x}}=k_{n-1}\tilde{x}^{(n-1)}
+\ldots+k_1\dot{\tilde{x}}+k_0\tilde{x}$, where $\mathbf{k}=[c_0\lambda^n, c_1\lambda^{n-1},\dots,c_{n-1}\lambda]$, $\lambda$ is 
a strictly positive constant and $c_i$ states for binomial coefficients, \textit{i.\,e}.\

\begin{equation}
c_i=\binom{n}{i}=\frac{n!}{(n-i)!\:i!}\:,\quad i=0,1,\ldots,n-1 
\label{eq:binom}
\end{equation}

Since in real-world applications the nonlinear system (\ref{eq:system}) is often not perfectly known, the control law 
(\ref{eq:lawfl}) based on conventional feedback linearization is not sufficient to ensure the exponential convergence 
of the tracking error to zero. 

Thus, we propose the adoption of fuzzy inference system within the control law, in order to compensate for $d$ 
and to enhance the feedback linearization controller.

\subsection{Fuzzy inference system}

Because of the possibility to express human experience in an algorithmic manner, fuzzy logic has been largely 
employed in the last decades to both control and identification of dynamical systems. 

\begin{center}
\textit{If $s$ is $S_r$ then $\hat{d}_r=\hat{D}_r$}\quad;\quad$r=1,2,\ldots,N$ 
\end{center}

\noindent
where $s=k_{n-1}\tilde{x}^{(n-1)}+\ldots+k_1\dot{\tilde{x}}+k_0\tilde{x}$ represents a combined tracking error measure,
$S_r$ are fuzzy sets, whose membership functions could be properly chosen, and $\hat{D}_r$ is the output value of each 
one of the $N$ fuzzy rules.

At this point, it should be highlighted that the adoption of a combined tracking error measure $s$ in the premise of 
the rules, instead of the state variables as in \citep{ijccc2013}, leads to a smaller number of fuzzy sets and rules, 
which simplifies the design process. Considering that external disturbances are independent of the state variables, 
the choice of a combined tracking error measure $s$ also seems to be more appropriate in this case.

Considering that each rule defines a numerical value as output $\hat{D}_r$, the final output $\hat{d}$ can be computed 
by a weighted average: 

\begin{equation}
\hat{d}(s) = \mathbf{\hat{D}}^{\mathrm{T}}\mathbf{\Psi}(s)
\label{eq:dcvector}
\end{equation}

\noindent
where, $\mathbf{\hat{D}}=[\hat{D}_1,\hat{D}_2,\dots,\hat{D}_N]$ is the vector containing the attributed values 
$\hat{D}_r$ to each rule $r$, $\mathbf{\Psi}(s)=[\psi_1(s),\psi_2(s),\dots,$ $\psi_N(s)]$ is a vector with components 
$\psi_r(s)= w_r/\sum_{r=1}^{N}w_r$ and $w_r$ is the firing strength of each rule.

\subsection{Fuzzy feedback linearization}

Considering that fuzzy logic can perform universal approximation \citep{kosko1}, we propose the adoption of a TSK 
fuzzy inference system within the feedback linearization controller to compensate for modeling inaccuracies and 
consequently enhance the trajectory tracking of uncertain nonlinear systems.

Therefore, the control law with the fuzzy compensation scheme can be stated as follows

\begin{equation}
u=b^{-1}[-f+x^{(n)}_d-k_0\tilde{x}-k_1\dot{\tilde{x}}-\cdots-k_{n-1}\tilde{x}^{(n-1)}-\hat{d}(s)]
\label{eq:lawfuzzy}
\end{equation}

\noindent
and the related closed-loop system is:

\begin{equation}
\tilde{x}^{(n)}+k_{n-1}\tilde{x}^{(n-1)}+\ldots+k_1\dot{\tilde{x}}+k_0\tilde{x}=\tilde{d}
\label{eq:clwd}
\end{equation}

\noindent
with $\tilde{d}=\hat{d}-d$.

\begin{Theorem}
\label{th:theo1}
Consider the uncertain nonlinear system (\ref{eq:system}) and Assumption~\ref{as:limd}, then the fuzzy feedback 
linearization controller defined by (\ref{eq:dcvector}) and (\ref{eq:lawfuzzy}) ensures the exponential convergence 
of the tracking error vector to a closed region $\Omega=\{\mathbf{x}\in\mathbb{R}^n\:|\:|\tilde{x}^{(i)}|\le\zeta_i
\lambda^{i-n}\varepsilon,i=0,1,\ldots,n-1\}$, with $\zeta_i$ defined by (\ref{eq:zeta}).

\begin{equation}
\zeta_i = \left\{\begin{array}{cl}
1&\mbox{for}\quad i=0 \\
1+\sum^{i-1}_{j=0}\binom{i}{j}\zeta_j&\mbox{for}\quad i=1,2,\ldots,n-1.
\end{array}\right.
\label{eq:zeta}
\end{equation}

\end{Theorem}

\begin{proof} Considering the universal approximation feature of fuzzy logic \citep{kosko1}, the output of the 
adopted inference system (\ref{eq:dcvector}) can approximate the disturbance $d$ to an arbitrary degree of 
accuracy, \textit{i.\,e}.\ $|\hat{d}(s)-d|\le\varepsilon$ for an arbitrary $\varepsilon>0$. Thus, 
from (\ref{eq:clwd}) one has

\begin{equation}
|\tilde{x}^{(n)}+k_{n-1}\tilde{x}^{(n-1)}+\ldots+k_1\dot{\tilde{x}}+k_0\tilde{x}|\le\varepsilon
\label{eq:bounds1}
\end{equation}

From (\ref{eq:binom}), inequality (\ref{eq:bounds1}) may be rewritten as

\begin{equation}
-\varepsilon\le \tilde{x}^{(n)}+c_{n-1}\lambda\tilde{x}^{(n-1)}+\cdots+c_1\lambda^{n-1}\dot{\tilde{x}}+c_0\lambda^n\tilde{x}\le\varepsilon
\label{eq:bounds2}
\end{equation}

Multiplying (\ref{eq:bounds2}) by $e^{\lambda t}$ yields

\begin{equation}
-\varepsilon e^{\lambda t}\le\frac{d^n}{dt^n}(\tilde{x}e^{\lambda t})\le\varepsilon e^{\lambda t} 
\label{eq:dbounds}
\end{equation}

Integrating (\ref{eq:dbounds}) between $0$ and $t$ gives

\begin{equation}
-\frac{\varepsilon}{\lambda} e^{\lambda t}+\frac{\varepsilon}{\lambda}\le\frac{d^{n-1}}{dt^{n-1}}(\tilde{x}
e^{\lambda t})-\left.\frac{d^{n-1}}{dt^{n-1}}(\tilde{x}e^{\lambda t})\right|_{t=0}\le\frac{\varepsilon}{\lambda} 
e^{\lambda t}-\frac{\varepsilon}{\lambda} 
\label{eq:int_1a}
\end{equation}

\noindent
or conveniently rewritten as

\begin{equation}
-\frac{\varepsilon}{\lambda} e^{\lambda t}-\left(\left|\frac{d^{n-1}}{dt^{n-1}}(\tilde{x}e^{\lambda t})
\right|_{t=0}+\frac{\varepsilon}{\lambda}\right)\le\frac{d^{n-1}}{dt^{n-1}}(\tilde{x}e^{\lambda t})\le
\frac{\varepsilon}{\lambda} e^{\lambda t}+\left(\left|\frac{d^{n-1}}{dt^{n-1}}(\tilde{x}e^{\lambda t})
\right|_{t=0}+\frac{\varepsilon}{\lambda}\right)
\label{eq:int_1b}
\end{equation}

The same reasoning can be repeatedly applied until the $n^\mathrm{th}$ integral of (\ref{eq:dbounds}) is reached:

\begin{multline}
-\frac{\varepsilon}{\lambda^n}e^{\lambda t}-\left(\left|\frac{d^{n-1}}{dt^{n-1}}(\tilde{x}e^{\lambda
t})\right|_{t=0}+\frac{\varepsilon}{\lambda}\right)\frac{t^{n-1}}{(n-1)!}-\cdots-\left(\left|\tilde{x}
(0)\right|+\frac{\varepsilon}{\lambda^n}\right)\le\tilde{x}e^{\lambda t}\le\frac{\varepsilon}{\lambda^n}
e^{\lambda t}+\\+\left(\left|\frac{d^{n-1}}{dt^{n-1}}(\tilde{x}e^{\lambda t})\right|_{t=0}
+\frac{\varepsilon}{\lambda}\right)\frac{t^{n-1}}{(n-1)!}+\cdots+\left(\left|\tilde{x}(0)\right|
+\frac{\varepsilon}{\lambda^n}\right)
\label{eq:int_n}
\end{multline}

Furthermore, dividing (\ref{eq:int_n}) by $e^{\lambda t}$, it can be easily verified that, for $t\to\infty$,

\begin{equation}
-\frac{\varepsilon}{\lambda^n}\le\tilde{x}(t)\le\frac{\varepsilon}{\lambda^n}
\label{eq:txbound}
\end{equation}

Considering the $(n-1)^\mathrm{th}$ integral of (\ref{eq:dbounds})

\begin{multline}
-\frac{\varepsilon}{\lambda^{n-1}}e^{\lambda t}-\left(\left|\frac{d^{n-1}}{dt^{n-1}}(\tilde{x}
e^{\lambda t})\right|_{t=0}+\frac{\varepsilon}{\lambda}\right)\frac{t^{n-2}}{(n-2)!}-\cdots
-\left(\left|\dot{\tilde{x}}(0)\right|+\frac{\varepsilon}{\lambda^{n-1}}\right)\le\frac{d}{dt}
(\tilde{x}e^{\lambda t})\le\frac{\varepsilon}{\lambda^{n-1}}e^{\lambda t}+\\+\left(\left|
\frac{d^{n-1}}{dt^{n-1}}(\tilde{x}e^{\lambda t})\right|_{t=0}+\frac{\varepsilon}{\lambda}\right)
\frac{t^{n-2}}{(n-2)!}+\cdots+\left(\left|\dot{\tilde{x}}(0)\right|+\frac{\varepsilon}{\lambda^{n-1}}\right)
\label{eq:int_n-2}
\end{multline}

\noindent
and noting that $d(\tilde{x}e^{\lambda t})/dt=\dot{\tilde{x}}e^{\lambda t}+\tilde{x}\lambda e^{\lambda
t}$, by imposing the bounds (\ref{eq:txbound}) to (\ref{eq:int_n-2}) and dividing again by $e^{\lambda
t}$, it follows that, for $t\to\infty$,
 
\begin{equation}
-2\frac{\varepsilon}{\lambda^{n-1}}\le\dot{\tilde{x}}(t)\le2\frac{\varepsilon}{\lambda^{n-1}}
\label{eq:txdbound}
\end{equation}

Now, applying the bounds (\ref{eq:txbound}) and (\ref{eq:txdbound}) to the $(n-2)^\mathrm{th}$ integral 
of (\ref{eq:dbounds}) and dividing once again by $e^{\lambda t}$, it follows that, for $t\to\infty$,

\begin{equation}
-6\frac{\varepsilon}{\lambda^{n-2}}\le\ddot{\tilde{x}}(t)\le6\frac{\varepsilon}{\lambda^{n-2}}
\label{eq:txddbound}
\end{equation}

The same procedure can be successively repeated until the bounds for $\tilde{x}^{(n-1)}$ are 
achieved:

\begin{equation}
-\left[1+\sum^{n-2}_{i=0}\binom{n-1}{i}\zeta_i\right]\frac{\varepsilon}{\lambda}\le\tilde{x}^{(n-1)}\le
\left[1+\sum^{n-2}_{i=0}\binom{n-1}{i}\zeta_i\right]\frac{\varepsilon}{\lambda}
\label{eq:txnbound}
\end{equation}

\noindent
where the coefficients $\zeta_i$ ($i=0,1,\ldots,n-2$) are related to the previously obtained bounds 
of each $\tilde{x}^{(i)}$ and can be summarized as in (\ref{eq:zeta}).

In this way, by inspection of the integrals of (\ref{eq:dbounds}), as well as (\ref{eq:txbound}), 
(\ref{eq:txdbound}), (\ref{eq:txddbound}), (\ref{eq:txnbound}) and the other omitted bounds, it follows 
that the tracking error exponentially converges to the $n$-dimensional box determined by the limits 
$|\tilde{x}^{(i)}|\le\zeta_i\lambda^{i-n}\varepsilon,i=0,1,\ldots,n-1$, where $\zeta_i$ is defined by 
(\ref{eq:zeta}). 
\end{proof}

\begin{Corollary}
\label{co:cor1}
It must be noted that the proposed control scheme provides a smaller tracking error when compared with 
the conventional feedback linearization controller. By setting the output of the fuzzy inference system 
to zero, $\hat{d}(\tilde{x})=0$, Theorem~\ref{th:theo1} implies that the resulting bounds are 
$|\tilde{x}^{(i)}|\le\zeta_i\lambda^{i-n}\delta,i=0,1,\ldots,n-1$. Considering that $\varepsilon<\delta$,
from the universal approximation feature of $\hat{d}$, it can be concluded that the tracking error obtained
with the fuzzy feedback linearization controller is smaller than the associated with the conventional scheme.
\end{Corollary}

\section{NUMERICAL SIMULATIONS}

Numerical simulations are now in focus exploring the controller capability to perform vibration reduction in 
smart structures. A fourth-order Runge-Kutta scheme is adopted. In all simulations, the material properties 
presented in Table~\ref{Tab1} are used. These values are chosen in order to match experimental data obtained 
by \citet{sittner1} for a Cu-Zn-Al-Ni alloy at 373 K (see Figure~\ref{Fig2}). For the data in Table~\ref{Tab1}, 
the parameters defined in equation~(\ref{eq:motion3}) assume the values:  $\alpha_2=1.240\times10^2$ and 
$\alpha_3=1.450\times10^4$. We further let $b=0.866$, corresponding to a two-bar truss with an initial position 
$\varphi_0=30$\textordmasculine. Sampling rates of $200\Omega/\pi$ for control system and $1000\Omega/\pi$ for 
dynamical model are assumed.

\begin{table}[hb]
\centering
\caption{Material Properties.}
\label{Tab1}
\begin{tabular}{lllll}
\hline\noalign{\smallskip}
$a_1$ (MPa/K) & $a_2$ (MPa) & $a_3$ (MPa) & $T_M$ (K) & $T_A$ (K)\\
\noalign{\smallskip}\hline\noalign{\smallskip}
523.29 & $1.868\times10^7$ & $2.186\times10^9$ & 288 & 364.3\\
\noalign{\smallskip}\hline
\end{tabular}
\end{table}

\begin{figure}[htb]
\centering
\includegraphics[width=0.4\textwidth]{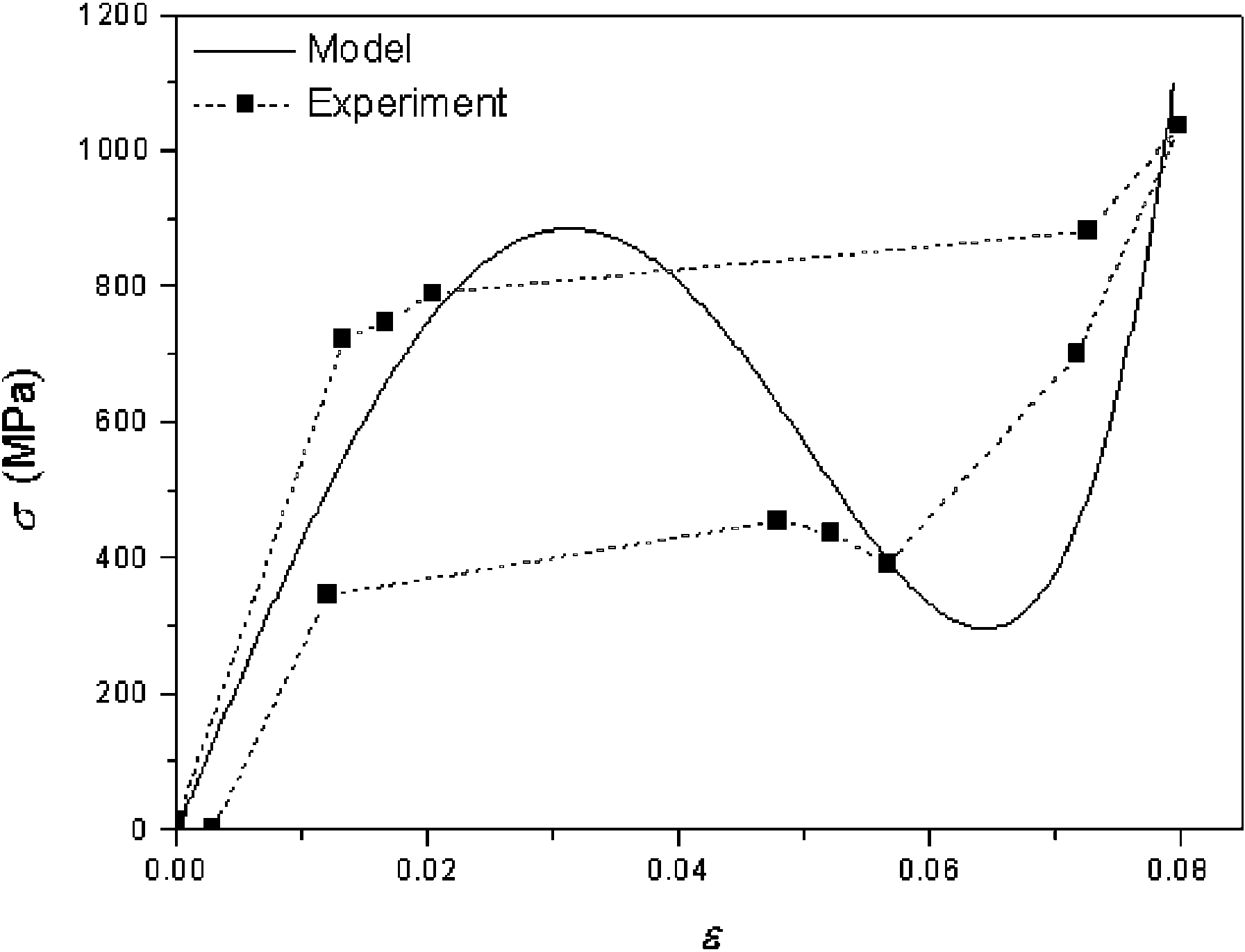}
\caption{Stress-strain curve: experimental and predicted by polynomial model.}
\label{Fig2}
\end{figure}

In order to demonstrate that the adopted control scheme can deal with both modeling inaccuracies and external disturbances, 
an uncertainty of $\pm20\%$ over the values of $\alpha_2$ and $\alpha_3$ is considered. Moreover, the periodic excitation 
is treated as an unknown external disturbance. Under this assumption, $\gamma\sin(\Omega\tau)$ is not taken into account 
within the design of the control law. On this basis, it is assumed estimation values as $\alpha_2=10^2$ and $\alpha_3=
1.15\times10^4$ in the control law. The other estimates in $f$ are chosen based on the assumption that model coefficients 
are perfectly known.  Concerning the fuzzy system, trapezoidal (at the borders) and triangular (in the middle) membership 
functions are adopted for $S_r$, with the central values defined respectively as $C=\{-1.00\:;\:$ $-0.50\:;\:-0.02\:;
\:0.02\:;\:0.50\:;\:1.0\}\times10^{-1}$. The control parameter $\lambda$ is defined as $\lambda=0.6$.

The stabilization of the state vector in the neighborhood of one of the equilibrium points of the shape memory 
two-bar truss \citep{savi5} is carried out. This approach shows that the adopted control scheme can significantly 
reduce the vibration level and also avoid the undesired snap-through behavior. Figures~\ref{Fig3} and~\ref{Fig4} 
show the obtained results considering $\mathbf{x}_d=[0.68,0.0]$, $\theta=0.69$, $\Omega=0.5$, $\xi=0.05$ and 
$\gamma=0.020$. Note that the chaotic behavior with large amplitudes of the uncontrolled respose is replaced by 
a regular behavior with small amplitudes around the equilibrium point. It should be highlighted that the proposed 
control law provides a smaller stabilization error, Figure~\ref{Fig4}, when compared with the conventional feedback 
linearization scheme, Figure~\ref{Fig3}.

\begin{figure}[htb]
\centering
\mbox{
\subfloat[Stabilization of $x$.]{\label{Fig3a}
\includegraphics[width=0.45\textwidth]{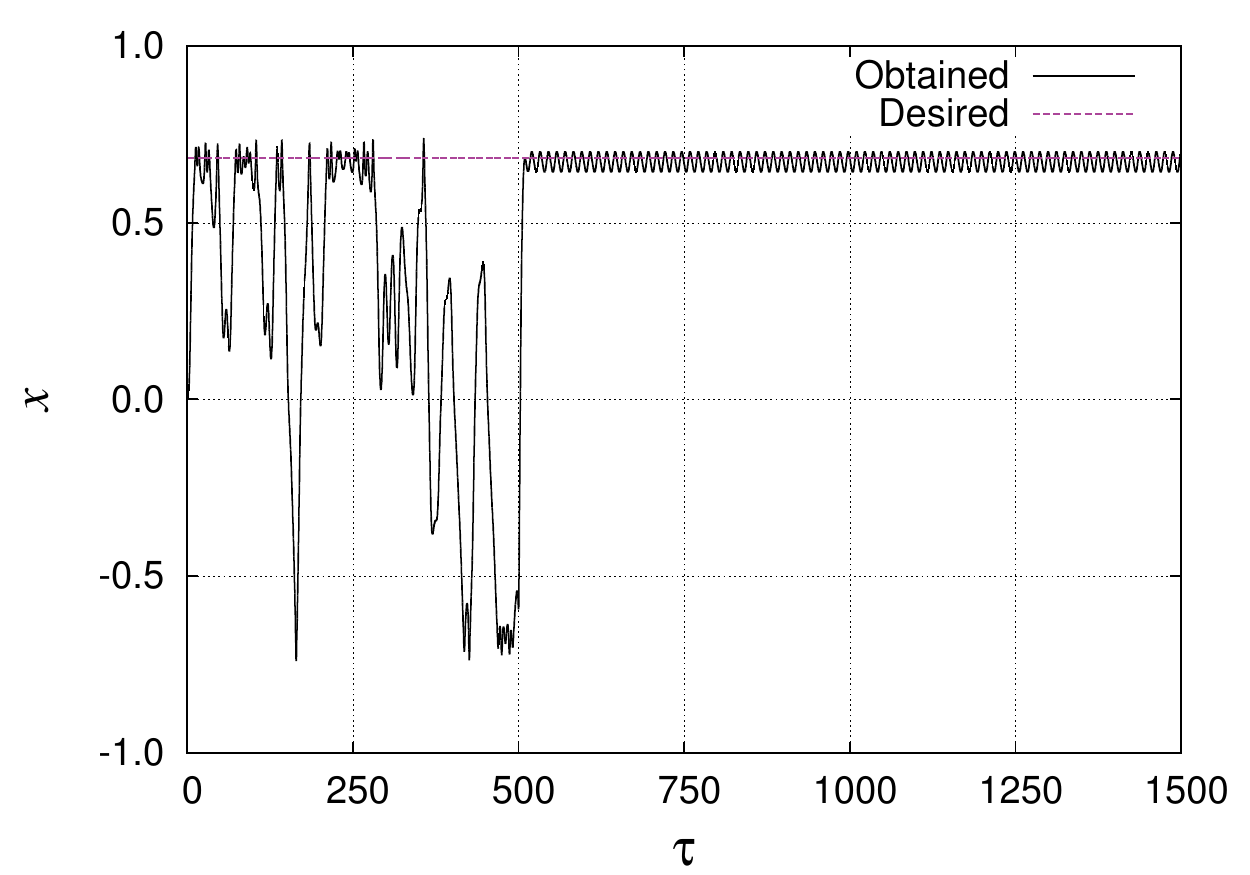}}
\subfloat[Control action $u$.]{\label{Fig3b}
\includegraphics[width=0.45\textwidth]{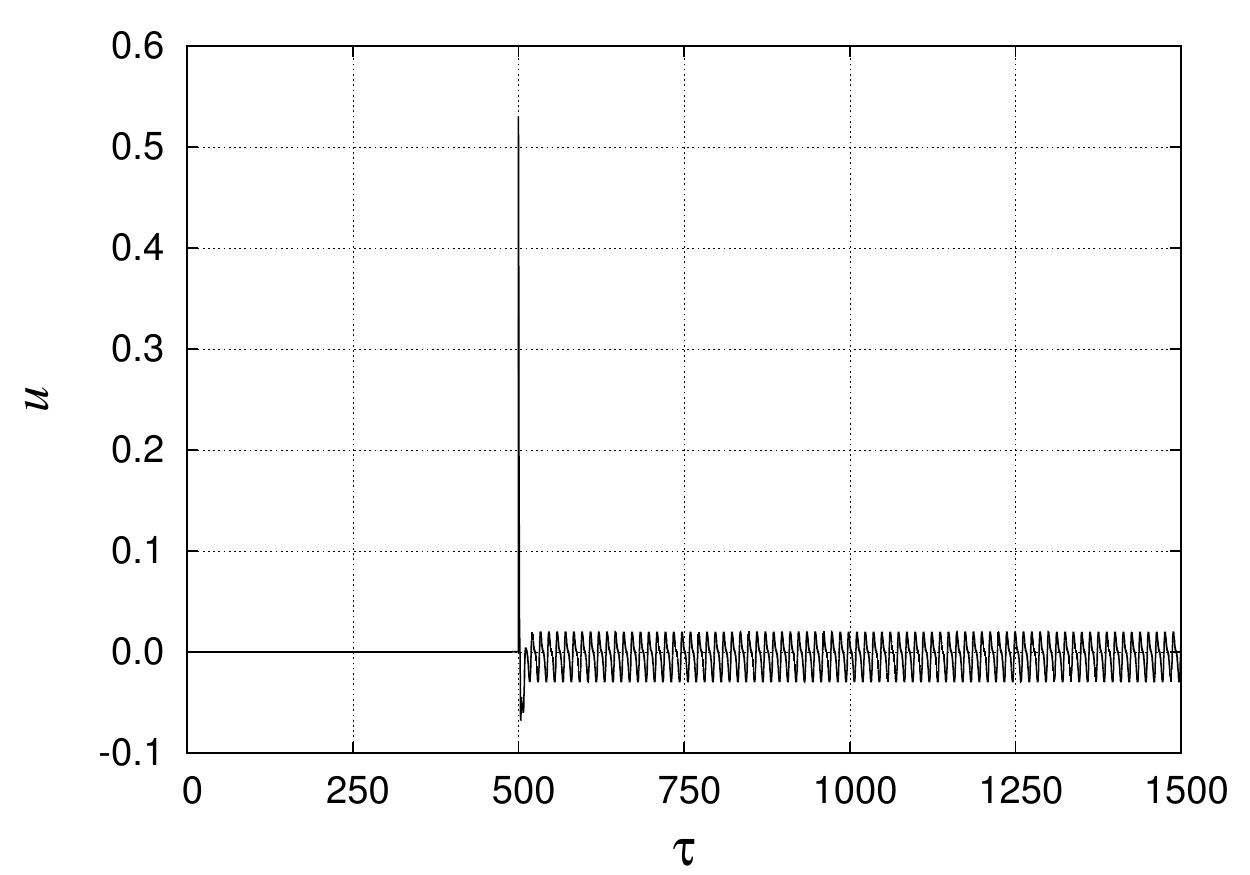}}
}
\caption{Controller performance without fuzzy compensation.}
\label{Fig3}
\end{figure}

\begin{figure}[htb]
\centering
\mbox{
\subfloat[Stabilization of $x$.]{\label{Fig4a}
\includegraphics[width=0.45\textwidth]{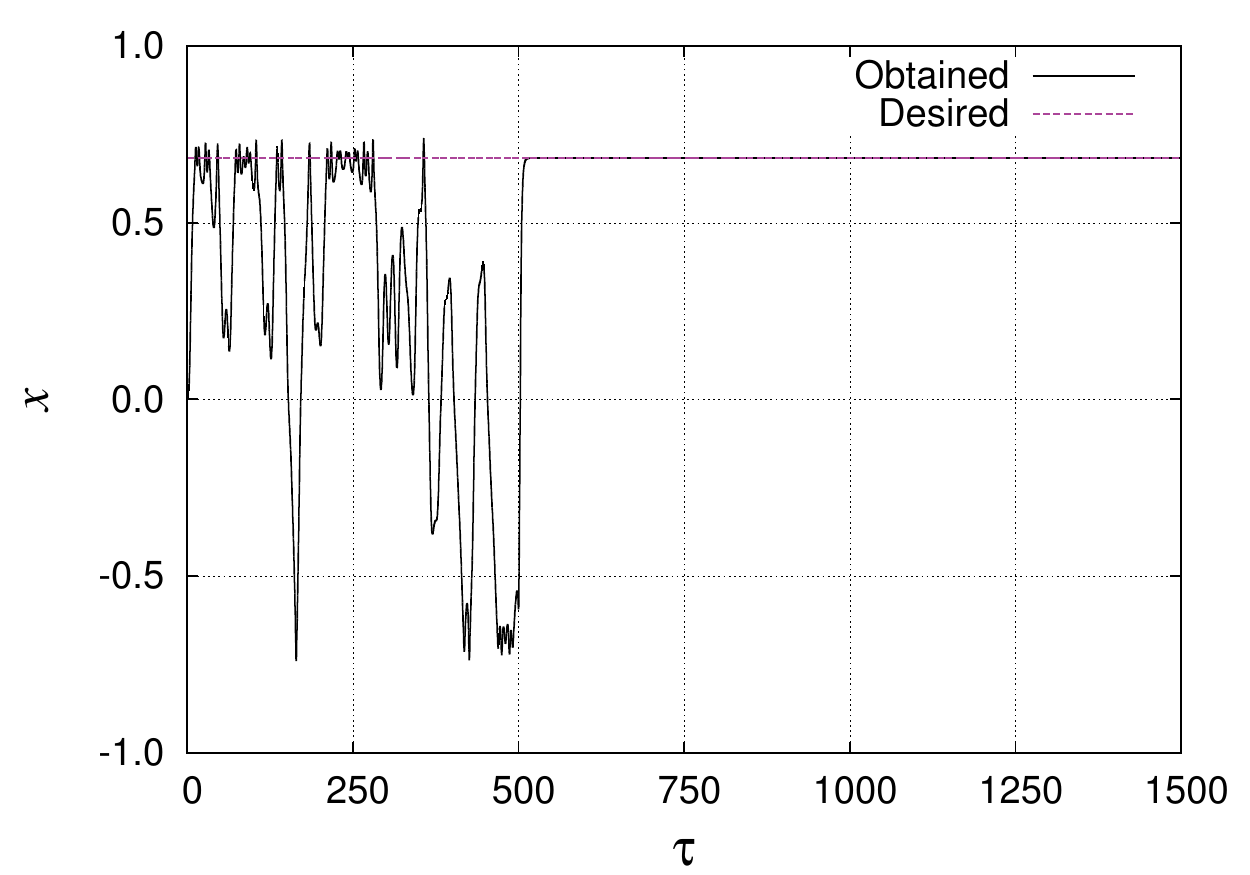}}
\subfloat[Control action $u$.]{\label{Fig4b}
\includegraphics[width=0.45\textwidth]{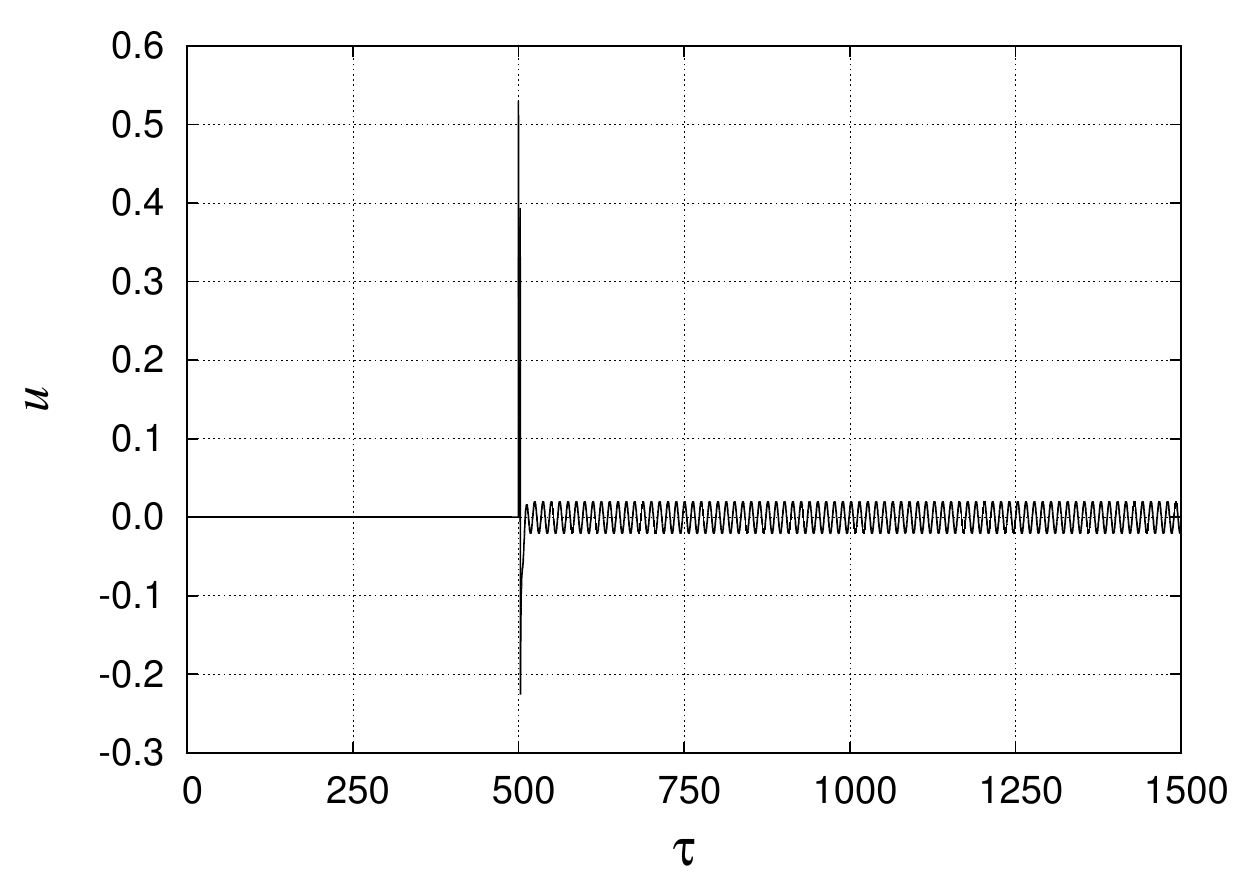}}
}
\caption{Controller performance with fuzzy compensation.}
\label{Fig4}
\end{figure}

\section{CONCLUDING REMARKS}

In this paper, a fuzzy feedback linearization controller is considered for vibration reduction in a shape memory 
two-bar truss. A polynomial constitutive model is assumed to describe the constitutive behavior of the bars. Despite 
the deceiving simplicity, this model allows an appropriate qualitative description of system dynamics, which can  
exhibit chaotic behavior. Numerical simulations show the efficacy of the proposed scheme against modeling inaccuracies 
and external disturbances. The improved performance over the conventional feedback linearization is also demonstrated. 
It should be highlighted that the controller robustness to modeling inaccuracies is an important issue that allows the 
use of a simple constitutive model for control purposes.

\section{ACKNOWLEDGEMENTS}

The authors would like to acknowledge the support of the Brazilian Research Agencies CNPq, CAPES and FAPERJ, 
and through the INCT-EIE (National Institute of Science and Technology - Smart Structures in Engineering)
the CNPq and FAPEMIG. The German Academic Exchange Service (DAAD) and the Air Force Office for Scientific 
Research (AFOSR) are also acknowledged.

\end{document}